\documentclass[a4paper,UKenglish,cleveref,thm-restate]{lipics-v2021}

\pdfoutput=1 
\hideLIPIcs
\nolinenumbers

\title{Degree-restricted strength decompositions and algebraic branching programs}

\author{Fulvio Gesmundo}{Saarland University, Saarbr\"ucken, Germany}{gesmundo@cs.uni-saarland.de}{https://orcid.org/
0000-0001-6402-021X}{}
\author{Purnata Ghosal}{University of Liverpool, Liverpool, UK}{purnata.ghosal@liverpool.ac.uk}{}{DFG IK 116/2-1 and EPSRC EP/W014882/1}
\author{Christian Ikenmeyer}{University of Liverpool, Liverpool, UK}{christian.ikenmeyer@liverpool.ac.uk}{}{DFG IK 116/2-1 and EPSRC EP/W014882/1}
\author{Vladimir Lysikov}{QMATH, Department of Mathematical Sciences, University of Copenhagen, Denmark}{vl@math.ku.dk}
{https://orcid.org/0000-0002-7816-6524}{ERC 818761 and VILLUM FONDEN via the QMATH Centre of Excellence (Grant No. 10059)}

\authorrunning{F. Gesmundo, P. Ghosal, C. Ikenmeyer and V. Lysikov}

\Copyright{Fulvio Gesmundo, Purnata Ghosal, Christian Ikenmeyer and Vladimir Lysikov} 

\acknowledgements{We would like to thank Daniele Agostini for many helpful discussions during the development of this work. We also thank Edoardo Ballico, Luca Chiantini, Giorgio Ottaviani and Kristian Ranestad for pointing out useful references on the Noether-Lefschetz property and related results.} 


\usepackage {tikz}
\usetikzlibrary{matrix,calc,backgrounds,arrows,shapes,patterns}
\usetikzlibrary{arrows,arrows.meta}
\usetikzlibrary{patterns,calc}
\usetikzlibrary{shapes}
\usetikzlibrary{shapes.geometric}

\newcommand{\bbF}{\mathbb{C}}
\newcommand{\bbZ}{\mathbb{Z}}

\newcommand{\abphom}{\operatorname{B}_{\mathrm{hom}}}
\newcommand{\str}{\operatorname{str}}
\newcommand{\sr}{\operatorname{sr}}

\newcommand{\id}{\operatorname{id}}

\newcommand{\bbA}{\mathbb{A}}
\newcommand{\bbC}{\mathbb{C}}
\newcommand{\bbP}{\mathbb{P}}
\newcommand{\Zero}{\operatorname{Z}}

\newcommand{\Sing}{\operatorname{Sing}}
\newcommand{\graph}{\operatorname{graph}}
\newcommand{\codim}{\operatorname{codim}}
\newcommand{\CH}{\operatorname{CH}}
\newcommand{\vvirg}{, \ldots ,}

\theoremstyle{remark}
\newtheorem{inclaim}{Claim}[theorem]

\begin{document}

\maketitle

\begin{abstract}
We analyze Kumar's recent quadratic algebraic branching program size lower bound proof method (CCC 2017) for the power sum polynomial.
We present a refinement of this method that gives better bounds in some cases.

The lower bound relies on Noether-Lefschetz type conditions on the hypersurface defined by the homogeneous polynomial. In the explicit example that we provide, the lower bound is proved resorting to classical intersection theory.

Furthermore, we use similar methods to improve the known lower bound methods for slice rank of polynomials.
We consider a sequence of polynomials that have been studied before by Shioda and show that for these polynomials the improved lower bound matches the known upper bound.
\end{abstract}

\section{Introduction}
Homogeneous algebraic branching programs are a fundamental machine model for the computation of homogeneous polynomials.
Their noncommutative version is completely understood (even in terms of border complexity) since Nisan's 1991 paper \cite{DBLP:conf/stoc/Nisan91};
they are the model of choice in \cite{burgisser2020rigid} for succinct presentation of a system of homogeneous polynomial equations;
and they can be used to phrase Valiant's famous determinant versus permanent question \cite{DBLP:conf/stoc/Valiant79a} in a homogeneous way:
Does the minimal size of the required homogeneous algebraic branching program for the permanent polynomial grow superpolynomially?
Phrasing Valiant's question in this way removes the ``padding'' problem in geometric complexity theory, so this is a way of circumventing the GCT no-go results in \cite{ikenmeyer2017rectangular,burgisser2019no}.

Even though Valiant's problem is the flagship problem in algebraic complexity theory, only very weak size bounds on homogeneous algebraic branching programs are known. The best lower bound so far was recently proved by Kumar~\cite{DBLP:journals/cc/Kumar19} for the power sum polynomial.
His proof and recent lower bounds proofs in other algebraic computation models (Chatterjee et~al.~\cite{DBLP:conf/coco/Chatterjee0SV20} for general algebraic branching programs, and Kumar and Volk~\cite{DBLP:conf/coco/0001V21} for determinantal complexity)
employ decompositions of the form
\begin{equation}
F = \sum_{k = 1}^r G_k H_k + R,
\end{equation}
where $F$ is the polynomial for which the lower bound is proven, and $G_k, H_k, R$ are polynomials such that $\deg R < \deg F$ and $G_k(0) = H_k(0) = 0$.
Kumar's recent $(d-1)\lceil\frac{n}{2}\rceil$ bound on the size of
homogeneous algebraic branching programs
can be proven by considering simpler decompositions
\begin{equation}\label{eq:strength-decomposition}
F = \sum_{k = 1}^r G_k H_k,
\end{equation}
where $F$ is a homogeneous polynomial and $G_k, H_k$ are homogeneous polynomials of degree strictly smaller than $F$. Decompositions of this form have been investigated in algebraic geometry as well: for instance, in~\cite{MR167458,MR781588}, they were used to study rational points on certain algebraic varieties; they appear in~\cite{MR2492444} to characterize complete intersections contained in a given hypersurface; recently, they were used in~\cite{MR4066476} to give a proof of Stillman's conjecture.
Following~\cite{MR4066476}, we say that~\eqref{eq:strength-decomposition} is a \emph{strength decomposition} of the polynomial $F$; the minimum $r$ for which a strength decomposition exists is called the \emph{strength} of $F$.
In the literature, the strength of $F$ is also called \emph{Schmidt rank} or \emph{$h$-invariant}.

In this paper we analyze Kumar's lower bound proof method.
The proof in~\cite{DBLP:journals/cc/Kumar19} is tailored to the power sum polynomial, but the method is applicable to every polynomial.
We present the general version of Kumar's technique in Section 3.1.
We refine this technique introducing the notion of $k$-restricted strength of a polynomial and we provide a lower bound for this notion based on geometric properties of the hypersurface defined by the polynomial.
These properties are based on the non-existence of subvarieties of low codimension and low degree in the associated hypersurface and can be interpreted as higher codimension versions of a Noether-Lefschetz type condition \cite{MR779603}.

We apply the refined method to give a lower bound on the size of a homogeneous algebraic branching program for an explicit family of polynomials
\[
  P_{n, d}(x_0, \dots, x_{2n}) = x_0^d + \sum_{k = 1}^n x_{2k - 1} x_{2k}^{d - 1};
\]
$P_{n, d}$ is a homogeneous polynomial of degree $d$ in $N = 2n + 1$ variables. 
Kumar's technique directly applied to this polynomial gives a lower bound $\lceil\frac{N}{4}\rceil (d - 1)$.
For $d < 2^{N/4}$ we improve the lower bound by an additive term of approximately ${N}/{2}$.
If the degree is exponential in~$N$, we get a further additive improvement of order $\frac{N}{2} d^{c/N}$, see~\Cref{cor:ourlowerbounds}(c).

In \Cref{section: slice rank}, we further study the notion of slice rank of homogeneous polynomials, which is a special case of $k$-restricted strength when $k = 1$.
\Cref{lem:sr-subspace-sing} gives a method to prove lower bounds on the slice rank.
Using this result, in \Cref{thm:exactslicerank} we compute the slice rank of polynomials
\[
S_{n, d}(x_0, \dots, x_{n+1}) = \sum_{i = 0}^{n - 1} x_i x_{i + 1}^{d - 1} + x_{n} x_0^{d - 1} + x_{n+1}^d.
\]
that have been studied by Shioda \cite{MR644520,MR726430}. 
For $n = 4$ and $6$ we find that the slice rank is equal to $\frac{N}{2} + 1$, where $N = n + 2$ is the number of variables.
To the authors' knowledge, this is the first lower bound for the slice rank better than $\lceil{N}/{2}\rceil$.
This translates to a lower bound $\frac{N}{2}(d - 1) + 2$ on the homoneous ABP complexity.
Again, this is the first lower bound better than Kumar's $\lceil\frac{N}{2}\rceil (d - 1)$.
We conjecture that these bounds continue to hold for $S_{n, d}$ with arbitrary even $n$, see \Cref{conj:lb}.

\section{Preliminaries}
\label{sec:prelim}

We work over the field of complex numbers $\mathbb{C}$. A homogeneous linear polynomial is called a \emph{linear form}. The ideal generated by polynomials $F_1, \dots, F_m$ is denoted by $\left< F_1, \dots, F_m \right>$. And ideal is homogeneous if it admits a set of homogeneous generators. Every (homogeneous) ideal in a polynomial ring admits a finite set of (homogeneous) generators \cite[Theorem 1.2]{Eis:CommutativeAlgebra}.

\subsection{Projective geometry}
\label{subsec:projvar}

Since we work with homogeneous polynomials, it is convenient to work in projective space $\bbP^{n}$, which is defined as $(\mathbb{C}^{n+1}\setminus\{0\})/\mathbb{C}^\times$, that is, points in $\bbP^n$ correspond to lines through the origin in $\bbC^{n + 1}$. Given a nonzero vector $v = (v_0 \vvirg v_n) \in \bbC^{n+1}$, write $[v] = (v_0 : \cdots : v_n)$ for the corresponding point in $\bbP^n$. We refer to \cite{Harris:AlgGeo} for basics of projective geometry and we only record some basic facts.
Given a homogeneous ideal $I = \langle F_1 \vvirg F_m\rangle$, write $\Zero(I) = \Zero(F_1 \vvirg F_m) = \{ [v] \in \bbP^n : F_j(v) = 0 \text{ for every $j$}\}$; a subset $X \subset \bbP^n$ is a variety if $X = \Zero(I)$ for some homogeneous ideal $I$. A variety is irreducible if it is not the proper union of two varieties; every variety $X$ can be written uniquely as finite union of finitely many irreducible subvarieties $X = \bigcup_1^r X_j$; $X_1 \vvirg X_r$ are called irreducible components of $X$.
We refer to \cite[Ch. 11 and Ch. 18]{Harris:AlgGeo} for the definitions and the basic properties of dimension and degree of a variety $X$, denoted respectively $\dim X$ and $\deg X$. The dimension of $\bbP^n$ is $n$. The codimension of $X \subseteq \bbP^n$ is $\codim X = n- \dim X$; if $X = \emptyset$, $\codim X = n+1$. A variety $X$ is called \emph{hypersurface} if all its irreducible components have codimension $1$; in this case $X = \Zero(F)$ is defined by a principal ideal $\langle F \rangle$.

A \emph{(projective) linear subspace} in $\bbP^n$ is a variety defined by linear forms. The codimension of $\Zero(L_1 \vvirg L_r)$, for some linear forms $L_1 \vvirg L_r$, equals the number of linearly independent elements among $L_1 \vvirg L_r$; in particular $\codim \Zero(L_1 \vvirg L_r) \leq r$.
A \emph{line} is a linear subspace of dimension $1$. The line spanned by two points $[x], [y] \in \bbP^{n}$ consists of all points of the form $[\alpha x + \beta y]$ with $(\alpha, \beta) \neq 0$.

Let $X \subset \bbP^n$ be a variety. The \emph{projective cone} over $X$ with vertex $p \notin X$ is the union of all lines connecting $p$ with a point in $X$.


Given a hypersurface $\Zero(F) \subseteq \bbP^n$, write $\Sing(F) = \Zero( \frac{\partial F}{\partial x_i} : i = 0 \vvirg n\}$, which is a subvariety of $\Zero(F)$. For example, if $F = x_0^d + x_1^d + x_2^d$, then $\Sing(F) = \Zero( x_0^{d-1} , x_1^{d-1}, x_2^{d-1} ) = \emptyset \subseteq \bbP^2$; if $F = x_0^d + x_1x_2^{d-1}$ then $\Sing(F) = \Zero(x_0^{d-1}, x_2^{d-1},x_1x_2^{d-2}) = \Zero(x_0,x_2) = \{ [0:1:0] \} \subseteq \bbP^2$. If the factorization of $F$ into irreducible polynomials does not have repeated factors, then $\Sing(F)$ coincides with the \emph{singular locus} of the hypersurface $\Zero(F)$, see,  e.g., in \cite[Ch. 14]{Harris:AlgGeo}.



We mention the following two fundamental results in algebraic geometry.

\begin{theorem}[Krull height theorem for polynomial rings {\cite[Cor. 11.17]{AtMac:CommutativeAlgebra}}]
\label{thm:Krullheight}
Let $X \subset \bbP^n$ be a variety, with $X = \Zero(F_1 \vvirg F_m)$. Then all irreducible components of $X$ have codimension at most $c$.
\end{theorem}

\begin{theorem}[Bezout inequality, see~{\cite[Thm.~8.28]{DBLP:books/daglib/0090316}}]
\label{thm:Bezout}
If a projective variety $X \subset \bbP^n$ is cut out by $c$ polynomials of degrees $d_1, \dots, d_c$, the sum of degrees of its irreducible components is at most $\prod_{i = 1}^c d_i$.
\end{theorem}

Further, in the proof of \Cref{lem:subvarieties-induction}, we require some basics of intersection theory. We introduce the required notions and references in \Cref{sec: lower bound degree restricted strength}.

\subsection{Algebraic branching programs}

The computational model of algebraic branching programs was first formally defined by Nisan~\cite{DBLP:conf/stoc/Nisan91} in the context of noncommutative computation, but essentially the same model was used by Valiant in his famous proof of universality of determinant~\cite{DBLP:conf/stoc/Valiant79a}.
The computational power of algebraic branching programs is intermediate between the one of general arithmetic circuits and the one of arithmetic formulas.
It is a convenient model for algebraic methods, because its power can be captured by restrictions of determinants or iterated matrix multiplication polynomials, which allows for the use of well developed tools from algebra and algebraic geometry.
In this paper we only consider homogeneous algebraic branching programs.

\begin{definition}
A \emph{layered directed graph} is a directed graph in which the set of vertices is partitioned into \emph{layers} indexed by integers so that each edge connects vertices in consecutive layers.
\end{definition}

\begin{definition}
A \emph{homogeneous algebraic branching program (ABP)} in variables $x_1, \dots, x_n$ is a layered directed graph with one source and one sink, and with edges labeled by linear forms in $x_1, \dots, x_n$.
The \emph{weight} of a path in an ABP is the product of labels on the edges of the path. The \emph{polynomial computed between vertices $u$ and $v$} is the sum of the weights of all paths from $u$ to $v$.
The \emph{polynomial computed by an ABP} is the polynomial computed between the source and the sink.

The \emph{size} of an ABP is the number of its inner vertices, namely all vertices except the source and the sink. 
For a homogeneous polynomial $F$, its \emph{homogeneous ABP complexity} $\abphom(F)$ is the minimal size of a homogeneous ABP computing $F$.
\end{definition}

See Figure~\ref{fig:ABP} for an example of an ABP.

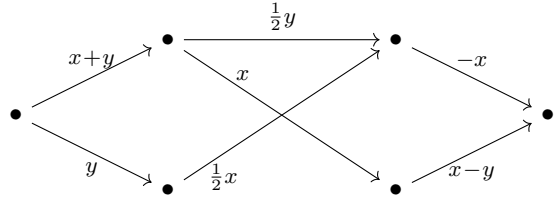
\begin{figure}
    \centering
\begin{tikzpicture}[xscale=2,yscale=1]
\node (0) at (0,0) {$\bullet$};
\node (1up) at (1,1) {$\bullet$};
\node (1down) at (1,-1) {$\bullet$};
\node (2up) at (2.5,1) {$\bullet$};
\node (2down) at (2.5,-1) {$\bullet$};
\node (3) at (3.5,0) {$\bullet$};
\draw[->] (0) to[above] node {\footnotesize $x\!+\!y$} (1up);
\draw[->] (0) to[below] node {\footnotesize $y$} (1down);
\draw[->] (1up) to[above] node {\footnotesize $\tfrac 1 2 y$} (2up);
\draw[->] (1up) to[pos=0.3, above] node {\footnotesize $x$} (2down);
\draw[->] (1down) to[pos=0.2, below] node {\footnotesize $\tfrac 1 2 x$} (2up);
\draw[->] (2up) to[above] node {\footnotesize $-x$} (3);
\draw[->] (2down) to[below] node {\footnotesize $x\!-\!y$} (3);
\end{tikzpicture}
    \caption{A homogeneous algebraic branching program of size 4, computing $(x+y) \cdot \frac 1 2 y \cdot (-x) + (x+y)\cdot x \cdot (x-y) + y\cdot \frac 1 2 x \cdot (-x) = x^3 - x^2 y - \frac{3}{2} x y^2$}
    \label{fig:ABP}
\end{figure}

\subsection{Strength and slice rank of polynomial}

\begin{definition}
Let $F$ be a homogeneous polynomial of degree $d$.
A \emph{strength decomposition} of $F$ is a decomposition of the form
\[
F = \sum_{k = 1}^r G_k H_k
\]
where $G_k$ and $H_k$ are homogeneous polynomials of degree less than $d$.
The strength of $F$ is 
\[
\str(F) = \min\{ r : F \text{ has a strength decomposition with $r$ summands}\}.
\]
\end{definition}

The following is a basic lower bound for the strength of a polynomial. It appears in the introduction of \cite{MR4066476} and \cite[Remark 4.3]{https://doi.org/10.48550/arxiv.2012.01237}; in~\cite{DBLP:journals/cc/Kumar19} it is mentioned with a reference to a personal communication with Saptharishi.

\begin{proposition}\label{prop:strength-codim-sing}
$\str(F) \geq \lceil \frac12 \codim \Sing(F) \rceil$.
\end{proposition}
\begin{proof}
If $F = \sum_{k = 1}^r G_k H_k$, then $\frac{\partial}{\partial x_i} F = \sum_{k = 1}^r G_k \frac{\partial}{\partial x_i} H_k + \sum_{k = 1}^r H_k \frac{\partial}{\partial x_i} G_k$.
Thus all the partial derivatives of $F$ lie in the ideal $\left< G_1, \dots, G_r, H_1, \dots, H_r\right>$ and therefore the zero set $\Zero(G_1, \dots, G_r, H_1, \dots, H_r)$ is contained in $\Sing(F)$.

Therefore, applying \Cref{thm:Krullheight}, we deduce
\[
\codim \Sing(F) \leq \codim \Zero(G_1, \dots, G_r, H_1, \dots, H_r) \leq 2r
\]
and the required lower bound follows.
\end{proof}

\begin{remark}
The bound of \Cref{prop:strength-codim-sing} gives essentially the only known lower bound method for strength which can be applied to explicit polynomials. Different methods are applied to polynomials of a specific form satisfying an unspecified genericity condition: for instance, in \cite{https://doi.org/10.48550/arxiv.2012.01237}, it is shown that a polynomial of the form $F = x_1^2 f_1 + x_2^2 f_2 + x_3^2 f_3 + x_4^2 f_4$ with generic $f_1 \vvirg f_4$ has strength $4$; this is however achieved using indirect methods 
\cite{Dra:TopNoetPolyFunct,Bik:Thesis}.
\end{remark}

\begin{definition}
Let $F$ be a homogeneous polynomial of degree $d$.
A \emph{slice rank decomposition} of $F$ is a decomposition of the form
\[
F = \sum_{k = 1}^r L_k H_k
\]
where $L_k$ are linear forms.
The minimal number of summands in a strength decomposition of $F$ is called the \emph{slice rank} of $F$ and is denoted by $\sr(F)$.
\end{definition}

We point out that the notion of slice rank of tensors \cite{SawTao:NoteSliceRank} is related but geometrically very different from the slice rank of homogeneous polynomials defined above. 

Clearly, slice rank decompositions are a special class of strength decompositions and thus $\sr(F) \geq \str(F)$. It is known that for generic polynomials the optimal strength decomposition is a slice rank decomposition \cite{https://doi.org/10.48550/arxiv.2102.11549}; in particular $\str(F) = \sr(F)$ for generic $F$.

Slice rank decompositions of a polynomial $F$ have a clear geometric interpretation in terms of linear subspaces contained in the hypersurface $\Zero(F)$.

\begin{proposition}\label{lem:sr-subspace}
Let $F$ be a homogeneous polynomial. We have $\sr(F) \leq r$ if and only if $\Zero(F)$ contains a linear subspace of codimension $r$.
\end{proposition}
\begin{proof}
The polynomial $F$ admits a decomposition $F = \sum_{k = 1}^r L_k H_k$ for linear forms $L_1 \vvirg L_k$ if and only if $F \in \langle L_1 \vvirg L_k \rangle$. The ideal $\langle L_1 \vvirg L_k \rangle$ is radical, in the sense of \cite[Sec. 1.6]{Eis:CommutativeAlgebra}. Therefore, by the classic Nullstellensatz \cite[Thm. 1.6]{Eis:CommutativeAlgebra}, the condition $F \in \langle L_1 \vvirg L_k \rangle$ is equivalent to the condition $\Zero(F) \supset \Zero(L_1, \dots, L_r)$.
\end{proof}

\section{Degree-restricted strength decompositions and lower bounds on ABP size}
\label{sec:degresstrdec}

\subsection{Basic properties and connection to ABPs}

In this section we introduce the degree-restricted strength decompositions and present a streamlined proof of Kumar's lower bound generalized to arbitrary polynomials based on \Cref{prop:strength-codim-sing}.

\begin{definition}
Let $F$ be a homogeneous polynomial of degree $d$.
A strength decomposition $F = \sum_{k = 1}^r G_k H_k$ is called \emph{$j$-restricted} if $\deg G_k = j$ for all $k$.
The \emph{$j$-restricted strength} $\str_j(F)$ is the minimal number of summands in a $j$-restricted strength decomposition of $F$.
\end{definition}

The following basic properties are clear from the definition.

\begin{proposition}\label{prop:restricted-strength}
Let $F$ be a homogeneous polynomial of degree $d$ and let $j$ be an integer such that $1 \leq j < d$.
The following statements hold.
\begin{alphaenumerate}
\item $\str_j(F) \geq \str(F)$;
\item $\str_j(F) = \str_{d - j}(f)$;
\item $\str_1(F) = \sr(F)$;
\end{alphaenumerate}
\end{proposition}

\begin{theorem}\label{thm:abp-sum}
For every homogeneous polynomial $F$ of degree $d$
\[
\abphom(F) \geq \sum_{j = 1}^{d - 1} \str_j(F).
\]
\end{theorem}
\begin{proof}
Let $\mathsf{A}$ be a homogeneous ABP computing $F$ with source $s$ and sink $t$.
Denote by $\mathsf{A}[v, w]$ the polynomial computed between vertices $v$ and $w$.
Let $V_j$ be the set of vertices in the $j$-th layer.
Since each path from the source to the sink contains exactly one vertex from each layer, we have
\[
F = \mathsf{A}[s, t] = \sum_{v \in V_j} \mathsf{A}[s, v] \mathsf{A}[v, t].
\]
If $v$ lies in the $j$-th layer then $\deg \mathsf{A}[s, v] = j$, because each path from $s$ to $v$ contains $j$ edges.
Thus $F$ has a $j$-restricted strength decomposition with $|V_j|$ summands, showing $|V_j| \geq \str_j(F)$.
Summing over all layers, we obtain the desired lower bound.
\end{proof}

From \Cref{thm:abp-sum}, \Cref{prop:restricted-strength}(a) and \Cref{prop:strength-codim-sing}, we immediately obtain the following $\abphom$ lower bounds technique.

\begin{corollary}[Kumar's singular locus lower bounds technique]\label{cor:bhombound}
For every homogeneous polynomial $F \in \bbF[x_0, \dots, x_n]$ of degree $d$
\[
\abphom(F)
\geq (d-1) \lceil \tfrac12 \codim \Sing(F) \rceil.
\]
In particular, if $\Sing(F) = \varnothing$, then 
\[
\abphom(F) \geq (d - 1) \lceil \tfrac{n+1}{2} \rceil.
\]
\end{corollary}

The power sum $F=x_0^d+\cdots+x_n^d$ of degree $d$ in $n+1$ variables satisfies $\Sing(F)=\varnothing$ and hence a direct application of \Cref{cor:bhombound} gives $\abphom(x_0^d+\cdots+x_n^d) \geq (d - 1) \lceil \frac{n+1}{2} \rceil$, which recovers Kumar's result \cite{DBLP:journals/cc/Kumar19}.

\subsection{Lower bound on degree-restricted strength}\label{sec: lower bound degree restricted strength}

In this section we prove a lower bound on the degree-restricted strength of polynomials.
It is based on the connection between strength decompositions and low degree subvarieties in $\Zero(F)$, which generalizes \Cref{lem:sr-subspace}.
We provide an explicit sequence of polynomials for which we obtain a lower bound that is slightly stronger than the one from \Cref{prop:strength-codim-sing}.

\begin{theorem}\label{thm:restricted-lowerbound}
Let $F \in \bbF[x_0, \dots, x_n]_d$ be a homogeneous polynomial.
Suppose that the zero set $\Zero(F)$ does not contain irreducible subvarieties $X$ with $\codim X \leq c$ and $\deg X < s$ for some $s \geq 2$.
Then
\[
\sr(F) \geq c + 1
\]
and
\[
\str_k(F) \geq \min\{c + 1, \lceil \log_k s \rceil\}
\]
\end{theorem}
\begin{proof}
The statement for the slice rank follows from~\Cref{lem:sr-subspace} as $\Zero(F)$ does not contain linear subspaces, that is, irreducible subvarieties of degree $1$,  of codimension $c$.

Assume $F$ has a $k$-restricted decomposition
\[
F = \sum_{j = 1}^r G_j H_j.
\]
Consider the variety $Y = \Zero(G_1, \dots, G_r)$. Since $F$ lies in the ideal $\langle G_1, \dots, G_r \rangle$, $Y$ is a subvariety of $\Zero(F)$. Since $Y$ is defined by $r$ polynomials of degree $k$, by \Cref{thm:Krullheight}, the codimension of every irreducible component of $Y$ is at most $r$; moreover, by \Cref{thm:Bezout} the sum of degrees of its irreducible components is at most $k^r$, hence the same holds for each component.

Suppose $r \leq c$. Since $\Zero(F)$ does not contain subvarieties $X$ with $\codim X \leq c$ and $\deg X < s$, for each irreducible component $X$ of $Y$ we have $k^r \geq \deg X \geq s$, so $r \geq \log_k s$. Since $r$ is an integer, we obtain the lower bound $r \geq \lceil \log_k s \rceil$. Hence, either $r \geq c+1$ or $r \geq\lceil \log_k s \rceil$, and we conclude $r \geq \min\{ c + 1, \lceil \log_k s \rceil\}$ as desired.
\end{proof}

\begin{remark}
In \Cref{thm:restricted-lowerbound} it suffices to require that $\Zero(F)$ does not contain subvarieties of codimension exactly $c$ and degree smaller than $s$. This is a consequence of Bertini's Theorem, see \cite[Sec.~18]{Harris:AlgGeo}. Indeed, if $X$ is an irreducible subvariety of $\Zero(F)$ with $\codim X < c$, let $X'$ be the intersection of $X$ with $c - \codim X$ generic hyperplanes; then $\codim X' = c$ and $\deg X = \deg X'$.
\end{remark}

We apply the lower bound of \Cref{thm:restricted-lowerbound} to the family of polynomials
\[
P_{n, d}(x_0, x_1, \dots, x_{2n}) = x_0^d + \sum_{k = 1}^n x_{2k - 1} x_{2k}^{d-1}
\]
with $d \geq 3$.
Note that it is clear from the definition of $P_{n, d}$ that $\str_k(P_{n, d}) \leq n + 1$.
Also note that $\Sing(P_{n, d})$ is the linear subspace given by $x_0 = x_2 = x_4 = \dots = x_{2n} = 0$.
Its codimension is $n + 1$, so the singular locus lower bound on $\abphom(P_{n,d})$ from \Cref{cor:bhombound} is
\begin{equation*}\label{eq:badbound}
\abphom(P_{n,d}) \geq (d-1)\left\lceil \tfrac{n + 1}{2} \right\rceil.
\end{equation*}

\Cref{cor:ourlowerbounds}(c) will provide an improvement of this lower bound, based on \Cref{thm:restricted-lowerbound}. In order to apply \Cref{thm:restricted-lowerbound}, we need a lower bound on the degree of subvarieties of $\Zero(P_{n,d})$ of low codimension. This is obtained resorting to an intersection theoretic argument, which is explained in the next section.

\subsection[Intersection theory of Z(P\_n,d)]{Intersection theory of $\Zero(P_{n,d})$}

This section requires some background in intersection theory, for which we refer to \cite{EisHar:3264} and \cite{Fulton}.
Given a variety $X$, denote by $\CH_a(X)$ the Chow group of algebraic cycle classes of dimension $a$ in $X$.
Elements of $\CH_a(X)$ are integer linear combinations of irreducible dimension $a$ subvarieties of $X$ modulo an equivalence relation called \emph{rational equivalence}.

If $X \subset Y$, then every subvariety of $X$ is also a subvariety of $Y$. This can be used to define a natural map $\iota_* : \CH_a(X) \to \CH_{a}(Y)$, the pushforward induced by the inclusion $X \subset Y$.
In particular, there is a well defined degree map $\deg : \CH_a(X) \to \bbZ$ sending an algebraic cycle class $Z \in \CH_a(X)$ to the degree of the class $\iota_*(Z) \in \CH_a(\bbP^n)$.

\begin{lemma}\label{lem:subvarieties-induction}
Let $F \in \bbF[x_0, \dots, x_n]_d$ be a homogeneous polynomial of degree $d \geq 2$. 
Set $G = F + x_{n + 1} x_{n + 2}^{d-1}$. Suppose that the degree of every subvariety $Y \subseteq \Zero(F) \subseteq \bbP^n$ with $\dim Y = a$ is divisible by $s$. Then the degree of every subvariety $Y' \subseteq \Zero(G) \subseteq \bbP^{n+2}$ with $\dim Y' = a + 1$ is divisible by $s$.
\end{lemma}
\begin{proof}
Let $X = \Zero(G) \subset \bbP^{n + 2}$ and $Z = \Zero(F) \subset \bbP^n$.
Let $Y \subset \bbP^{n + 1}$ be the variety given by the equation $F(x_0, \dots, x_n) = 0$ in $\bbP^{n + 1}$.
It consists of all points of the form $[\alpha x + \beta e_{n + 1}]$ with $F(x) = 0$, so it is the projective cone over $Z$ with the vertex $[e_{n + 1}]$.

It is known~\cite[Ex.~2.6.2]{Fulton} that the Chow group $\CH_{a + 1}(Y)$ is isomorphic to the Chow group $\CH_a(Z)$.
The isomorphism $\alpha \colon \CH_a(Z) \to \CH_{a + 1}(Y)$ takes the class of a subvariety in $Z$ to the class of the cone over this subvariety.

Let $H$ be the hyperplane given by $x_{n + 2} = 0$ in $\bbP^{n + 2}$.
Note that $X \cap H \subset H$ is isomorphic to $Y \subset \bbP^{n + 1}$, so we can identify $X \cap H$ with~$Y$.
Let $U$ be the open subset $X \setminus Y$.
This subset is an affine variety in $\bbA^{n + 2} = \bbP^{n + 2} \setminus H$ given by the equation $F + x_{n + 1} = 0$.
Thus $U \cong \graph F \cong \bbA^{n + 1}$. It follows that the Chow group $\CH_{a + 1}(U)$ is trivial~\cite[\S 1.9]{Fulton}.

The exactness of the excision exact sequence~\cite[Prop.~1.8]{Fulton}
\[
\CH_{a + 1}(Y) \to \CH_{a + 1}(X) \to \CH_{a + 1}(U) \to 0
\]
implies that the inclusion pushforward $\iota^* \colon \CH_{a + 1}(Y) \to \CH_{a + 1}(X)$ is surjective.

Both the cone map $\alpha \colon \CH_a(Z) \to \CH_{a + 1}(Y)$ and $\iota^* \colon \CH_{a + 1}(Y) \to \CH_{a + 1}(X)$ preserve the degree.
Composing them, we obtain a degree-preserving surjective map from $\CH_a(Z)$ to $\CH_{a + 1}(X)$.
Since the degree of every dimension $a$ subvariety of $Z$ is divisible by $s$, the same is true for cycle classes in $\CH_a(Z)$ and therefore, for cycle classes in $\CH_{a + 1}(X)$, including dimension $a + 1$ subvarieties of $X$.
\end{proof}

\begin{corollary}\label{cor:ourlowerbounds}
If $n \geq 1$ and $d \geq 2$, then 

\begin{alphaenumerate}
\item $\sr(P_{n, d}) = n + 1$;
\item $\str_k(P_{n, d}) \geq \min \{n + 1, \lceil \log_k d \rceil\}$;
\item $\abphom(P_{n, d}) \geq (d - 1) \lceil \frac{n + 1}{2} \rceil + 2 \lfloor \frac{n + 1}{2} \rfloor$ for $d \leq 2^{\frac{n + 1}{2}}$;\\
$\abphom(P_{n, d}) \geq (d - 1) \lceil \frac{n + 1}{2} \rceil
+ 2 \lfloor \frac{n + 1}{2} \rfloor \lfloor d^{\frac{1}{n + 1}} \rfloor + \sum_{j = \lfloor d^{\frac{1}{n + 1}} \rfloor + 1}^{\lfloor d^{\frac{2}{n + 1}} \rfloor} (\lceil \log_j d \rceil - \lceil \frac{n + 1}{2} \rceil)$ for all $d$.
\end{alphaenumerate}

\end{corollary}
\begin{proof}
The polynomial $P_{1, d} = x_0^d + x_1 x_2^{d - 1}$ is irreducible, e.g., by the Eisenstein criterion~\cite[Ex.~18.11]{Eis:CommutativeAlgebra} applied to it as an element of $\bbF[x_1, x_2][x_0]$ with the prime ideal $\left< x_1 \right>$.

It follows that $\Zero(P_{1, d})$ is an irreducible variety of degree $d$ in $\bbP^2$, so it does not contain subvarieties of codimension $1$ other than itself; in particular the degree every subvariety of codimension $1$ has degree divisible by $d$.
Using \Cref{lem:subvarieties-induction} inductively, we obtain that every subvariety $X$ of $\Zero(P_{n, d})$ with $\codim X = n$ has degree divisible by $d$. The lower bounds for the slice rank and the restricted strength follow by \Cref{thm:restricted-lowerbound}. 

The lower bound for the homogeneous ABP complexity follows by \Cref{prop:restricted-strength}(b) and \Cref{thm:abp-sum}.
If $d \leq 2^{\frac{n + 1}{2}}$, we use $\str_1(P_{n, d}) = \str_{d - 1}(P_{n, d}) = n + 1$ and $\str_j(P_{n, d}) \geq \lceil \frac{n + 1}{2} \rceil$ from \Cref{prop:strength-codim-sing} for other $j$ to get
\[
\abphom(P_{n, d}) \geq \sum_{j = 1}^{d - 1} \str_j(P_{n, d}) \geq 2(n + 1) + (d - 3) \lceil \frac{n + 1}{2} \rceil = (d - 1) \lceil \frac{n + 1}{2} \rceil + 2 \lfloor \frac{n + 1}{2} \rfloor
\]

For the second bound in (c), we separate the sum $\sum_{j = 1}^{d - 1} \str_j(P_{n, d})$ into three parts.

For $j \leq d^{\frac{1}{n + 1}}$ we have $\str_j(P_{n, d}) = \str_{d - j}(P_{n, d}) \geq n + 1$,
for $d^{\frac{1}{n + 1}} < j \leq d^{\frac{2}{n + 2}}$ we use the lower bound $\str_j(P_{n, d}) = \str_{d - j}(P_{n, d}) \geq \lceil \log_j d \rceil$, and for $d^{\frac{2}{n + 2}} < j < d - d^{\frac{2}{n + 2}}$ \Cref{prop:strength-codim-sing} gives $\str_j(P_{n, d}) \geq \lceil \frac{n + 1}{2} \rceil$.
\end{proof}

We point out that determining explicit hypersurfaces which do not contain low codimension subvarieties of low degree is an extremely hard problem. It is related to the Noether-Lefschetz Theorem, a classical result which, in particular, implies that if $F$ is a general homogeneous polynomial of degree $d$, then $\Zero(F)$ does not contain subvarieties $X$ with $\codim X = 2$ and $\deg X \leq d$. A consequence of \cite{Wu:ConjGrifHarNoetherLefschetz} is that if $F$ is a general homogeneous polynomial of degree $d \geq 6$ in five variables, then $\Zero(F) \subseteq \bbP^4$ does not contain subvarieties $X$ with $\codim X \leq 3$  and $\deg(X) \leq d$. Stronger cohomological results hold for very general hypersurfaces; a series of conjectures and open problems is proposed in~\cite{MR779603}.

\subsection{Slice rank lower bound}\label{section: slice rank}

In this section, we give examples of polynomials for which we prove a slice rank lower bound stronger than the one induced by the strength lower bound of \Cref{prop:strength-codim-sing}. In one instance, we show an improved lower bound for a polynomial with $\Sing(F) = \varnothing$.

We define the \emph{Shioda polynomials} of degree $d \geq 3$ in $n+2$ variables:
\[
S_{n, d}(x_0, \dots, x_{n+1}) = \sum_{i = 0}^{n - 1} x_i x_{i + 1}^{d - 1} + x_{n} x_0^{d - 1} + x_{n+1}^d.
\]
The polynomials $S_{n, d}$ were investigated by Shioda~\cite{MR644520,MR726430} as explicit examples for a cohomological version of the Noether-Lefschetz Theorem in middle dimension.

For even $n$ the decomposition
\begin{equation}\label{eq:shioda-decomp}
S_{n, d} = \sum_{k = 0}^{n/2-1} x_{2k+1}(x_{2k+2}^{d - 1} + x_{2k} x_{2k+1}^{d - 2}) + x_n \cdot x_0^{d - 1} + x_{n + 1} \cdot x_{n + 1}^{d - 1}
\end{equation}
shows that $\sr(S_{n, d}) \leq \frac{n}{2} + 2$.
We prove the matching lower bounds for $n = 2$ and $n = 4$; we conjecture that the bound holds for all even values of $n$.

First consider a modified polynomial
\[
\hat{S}_{n, d}(x_1, \dots, x_{n+1}) = \sum_{i = 1}^{n - 1} x_i x_{i + 1}^{d - 1} + x_{n+1}^d,
\]
which is obtained from $S_{n, d}$ by setting $x_0 = 0$. It is easy to verify that $\Sing(\hat{S}_{n, d})$ coincides with the point $[e_{n-1}] = (0:\dots:1:0:0)$; here we use homogeneous coordinates $(x_1 : \cdots : x_{n+1})$ on $\bbP^n$.

If $n$ is even, then $\codim \Sing(\hat{S}_{n,d}) = n$ and \Cref{prop:strength-codim-sing} gives a lower bound $\sr (\hat{S}_{n, d}) \geq \frac{n}{2}$. An analog of \eqref{eq:shioda-decomp} gives the upper bound $\frac{n}{2} + 1$. We provide a matching lower bound, relying on the following result which improves the lower bound of \Cref{prop:strength-codim-sing}.

\begin{theorem}\label{lem:sr-subspace-sing}
Let $F \in \bbF[x_0, \dots, x_n]$ be a homogeneous polynomial with $\codim \Sing(F) = s$ even.
Then $\sr(F) = \frac{s}{2}$ if and only if There is a linear space $Q \subset \Zero(F)$ of codimension $\frac{s}{2}$ containing one of the irreducible components of $\Sing(F)$. 
\end{theorem}
\begin{proof}
If $Q \subset \Zero(F)$ is a linear space with $\codim Q = s/2$, then $\sr(F) = s/2$ by \Cref{lem:sr-subspace}.

Conversely, suppose $\sr(F) = s/2$ and let $F = \sum_{k = 1}^{s/2} L_k H_k$ be a minimal slice rank decomposition of $F$.
Let $Q = \Zero( L_1, \dots, L_{s/2})$, and let $X = \Zero(L_1, \dots, L_{s/2}, H_1, \dots, H_{s/2})$. By the minimality of the decomposition, $Q$ is a linear space of codimension $s/2$. Moreover, $X \subseteq Q \subseteq \Zero(F)$.

Since $X$ is defined by $s$ polynomials, by \Cref{thm:Krullheight} all irreducible components of $X$ have codimension at most $s$. As in \Cref{prop:strength-codim-sing}, $X$ is contained in $\Sing(F)$. Therefore, for every irreducible component $X'$ of $X$, we have $s = \codim \Sing(F) \leq \codim X \leq \codim X' \leq s$. This shows that $X'$ and $\Sing(F)$ have the same dimension, therefore $X'$ is an irreducible component of $\Sing(F)$. Since $X' \subseteq X \subseteq Q$, we conclude. 
\end{proof}

\begin{lemma}\label{lem:reduced-shioda-lowerbound}
If $d \geq 3$ and $n$ is even, then $\sr(\hat{S}_{n,d}) = \frac{n}{2} + 1$
\end{lemma}
\begin{proof}
If $n$ is even, \cref{prop:strength-codim-sing} gives the lower bound $\sr(\hat{S}_{n,d}) \geq \frac{n}{2}$. We use induction on $n$ to prove that $\sr(\hat{S}_{n,d}) \neq \frac{n}{2}$. If $n = 0$ the statement is clear.

Suppose by contradiction $\sr(\hat{S}_{n, d}) = \frac{n}{2}$. By \cref{lem:sr-subspace-sing} there exists a projective linear subspace $Q \subset \Zero(\hat{S}_{n, d})$ of dimension $\frac{n}{2}$ containing the singular point $[e_{n-1}]$. Let $H = \Zero(x_{n-1})$ and let $Q' = Q \cap H$. Since $Q$ contains the point $[e_{n - 1}]$, which is not in $H$, we have $\dim Q' = \frac{n}{2} - 1$.

Observe that for every point $(v_1 : \cdots : v_n) \in Q'$, we have $v_{n} = 0$. To see this, fix $[v] \in Q'$ and consider the line spanned by $[v] \in Q'$ and $[e_{n-1}]$; this line is contained in $Q$, hence in $\Zero(\hat{S}_{n, d})$. In particular, for every $\alpha,\beta \in \bbC$, we have
\[
0 = \hat{S}_{n, d}(\alpha v + \beta e_{n-1}) = \alpha^d (\sum_{k = 1}^{n - 3} v_k v_{k + 1}^{d - 1} + v_{n + 1}^d) + \alpha^{d - 1}\beta v_{n}^{d - 1}.
\]
Therefore, this expression must be $0$ as a polynomial in $\alpha,\beta$ and in particular $v_n = 0$ because $v_{n}^{d - 1}$ is the coefficient of $\alpha^{d - 1}\beta$.

This shows that the existence of a subspace $Q \subset \Zero(\hat{S}_{n, d})$ of dimension $\frac{n}{2}$ containing $[e_{n-1}]$ implies the existence of a subspace $Q' \subset \Zero(\hat{S}_{n, d}) \cap \Zero(x_{n - 1}, x_n)$ of dimension $\frac{n}{2} - 1$. Note that substituting $x_{n - 1} = x_n = 0$ into $\hat{S}_{n, d}$, one obtains, up to renaming the variables, the polynomial $\hat{S}_{n - 2, d}$. Therefore, the existence of $Q'$ implies $\sr(\hat{S}_{n - 2, d}) = \frac{n}{2} - 1 = \frac{n - 2}{2}$, in contradiction with the induction hypothesis. This concludes the proof. 
\end{proof}

\begin{theorem}\label{thm:exactslicerank}
If $d \geq 5$, then $\sr(S_{4, d}) = 4$.
\end{theorem}
\begin{proof}
Let $\rho \colon \bbP^5 \to \bbP^5$ be the map defined by $\rho(x_0 : x_1 : \dots :x_5) = (x_4 : x_0 : x_1 : x_2 : x_3 :x_5 )$  which cyclically permutes the first $5$ coordinates of a point $(x_0:x_1:\dots:x_5)$.
Note that the hypersurface $\Zero(S_{4, d})$ is mapped to itself by $\rho$.

The lower bound given by \Cref{prop:strength-codim-sing} is $\sr(S_{4,d}) \geq 3$ and assume by contradiction that equality holds. Then there exists a linear space $Q \subset \Zero(S_{4, d}) \subset \bbP^5$ of codimension $3$.

First, we prove a series of claims about the plane $Q$ culminating with the claim that $Q$ contains one of the five points $[e_k] = \rho^k [e_0]$, where $e_0 = (1,0,\dots,0)$. Then we derive a contradiction with the lower bound for $\hat{S}_{n, d}$.

Let $A_0$ be the line $\Zero(x_1, x_3, x_4, x_5) \subseteq \bbP^5$ and $A_k = \rho^k A_0$. In other words, $A_0$ is the set of points of the form $(x_0:0:x_2:0:0:0)$ and $A_k$ is obtained by cyclically shifting the first five coordinates.

\begin{inclaim}\label{claim:triple}
$Q$ intersects $A_0 \cup A_1 \cup A_2$.
\end{inclaim}
\begin{claimproof}
Since $\dim Q =2$, by \cite[Prop. 11.4]{Harris:AlgGeo} it intersects any codimension $2$ linear subspace. Let $p = (p_0:0:p_2:0:p_4:p_5) \in Q \cap \Zero(x_1, x_3)$ and $q = (q_0:q_1:0:q_3:0:q_5) \in Q \cap \Zero(x_2, x_4)$ be two points which lie in the respective intersections.

If $p = q$, then $p = q = [e_0] \in A_0$ and the Claim is verified.

Suppose $p \neq q$. The line joining $p$ and $q$ lies in $Q$ and, therefore, in $\Zero(S_{4, d})$.
Let $\hat{p}, \hat{q} \in \bbC^{6}$ be representatives of $p$ and $q$ respectively. We have $S_{4, d}(\alpha \hat{p} + \beta \hat{q}) = 0$ for all values of $\alpha$ and $\beta$. Expand this expression as a polynomial in $\alpha,\beta$, and consider the coefficients of $\alpha^d$, $\alpha^{d - 2}\beta^2$, $\alpha^{d - 3} \beta^3$; these must be $0$, hence 
\begin{align*}
&    p_4 p_0^{d - 1} + p_5^d = 0 ,\\
&    \binom{d-1}{2} p_4 p_0^{d - 3} q_0^2 + \binom{d}{2} p_5^{d - 2} q_5^2 = 0 ,\\
&    \binom{d-1}{3} p_4 p_0^{d - 4} q_0^3 + \binom{d}{3} p_5^{d - 3} q_5^3 = 0.
\end{align*}
Rewrite these equations as
\begin{align*}
    &p_4 p_0^{d - 1} = - p_5^d ,\\
    &(d - 2) p_4 p_0^{d - 3} q_0^2 = -d p_5^{d - 2} q_5^2, \\
    &(d - 3) p_4 p_0^{d - 4} q_0^3 = -d p_5^{d - 3} q_5^3.
\end{align*}
If $p_5 \neq 0$, then dividing the equations by $p_4 p_0^{d - 1} = -p_5^d$, we obtain
\begin{align*}
    (d - 2) \left(\frac{q_0}{p_0}\right)^2 = d \left(\frac{q_5}{p_5}\right)^2 \\
    (d - 3) \left(\frac{q_0}{p_0}\right)^3 = d \left(\frac{q_5}{p_5}\right)^3
\end{align*}
from which we have $q_0 = q_5 = 0$, which implies $q \in A_1$.
If, on the other hand, $p_5 = 0$, then either $p_0 = 0$ and $p \in A_2$, or $p_4 = 0$ and $p \in A_0$.
\end{claimproof}

\begin{inclaim}\label{claim:pair}
If $Q \cap A_0 \neq \varnothing$ and $Q \cap A_2 \neq \varnothing$, then $Q$ contains $[e_0]$, $[e_2]$, or $[e_4]$.
\end{inclaim}
\begin{claimproof}
Let $p$ be a point in $Q \cap A_0$ and $q \in Q \cap A_2$, so 
$p = (p_0:0:p_2:0:0:0)$ and $q = (0:0:q_2:0:q_4:0)$.
If $p = q$, then they are equal to $[e_2]$.
If $p \neq q$, then there is a point on the line that they span which has zero second coordinate.
Let this linear combination be $r = (r_0:0:0:0:r_4:0)$.
Since $r \in Q \subset \Zero(S_{4, n})$, we have $S_{4, n}(r) = r_4 r_0^{d - 1} = 0$, so $r$ is either $[e_0]$ or $[e_4]$.
\end{claimproof}

\begin{inclaim}
$Q$ contains one of the five points $e_k$.
\end{inclaim}
\begin{claimproof}
Let $S = \{ k \mid Q \cap \rho^k A_0 \neq 0 \}$. Since $\rho^5 = \id$, we can see $S$ as a subset of $\bbZ / 5 \bbZ$.
\Cref{claim:triple} implies that $S$ contains at least one element of $\{0, 1, 2\}$.
Because of the cyclic symmetry, an analogous statement is true for every three consecutive values in $\bbZ/5\bbZ$.
Similarly, cyclically shifted versions of \Cref{claim:pair} imply that if $S$ contains $k$ and $k + 2$, then $Q$ contains one of the five basis points.

Without loss of generality, $0 \in S$.
If $2 \in S$ or $3 \in S$, then \Cref{claim:pair} guarantees $Q$ contains one of the five basis points. If both $2, 3 \notin S$, then \Cref{claim:triple} applied to the consecutive triples $\{1,2,3\}$ and $\{2,3,4\}$ implies that $1,4 \in S$. Since they differ by two, \Cref{claim:pair} applies and we conclude.
\end{claimproof}

By shifting $Q$ cyclically we can assume that $[e_1] \in Q$.

\begin{inclaim}
If $e_1 \in Q$, then $Q$ lies in the hyperplane $\Zero(x_0)$.
\end{inclaim}
\begin{claimproof}
Note that if the line spanned by two points $[v],[w]$ lies in a hypersurface $\Zero(F)$, then $\sum \frac{\partial F}{ \partial x_i} (v) w_j = 0$. Indeed, the function $f(t) = F(\alpha  + t w)$ is identically $0$ and so is its derivative in $t$ at $t=0$. By the chain rule, we obtained the desired equality. In particular, for every two points $p,q \in Q$, we obtain $\sum \frac{\partial S_{4, n}}{\partial x_i}(p) q_i = 0$. Fixing $p = [e_1]$, this guarantees $q_0 = 0$ for every $q \in Q$, showing $Q \subseteq \Zero(x_0)$.
\end{claimproof}

We deduce that $Q$ is contained in $\Zero(S_{4, n}) \cap \Zero(x_0) = \Zero (S_{4,n},x_0) = \Zero (\hat{S}_{4,n},x_0)$. Therefore $Q\subset \Zero (\hat{S}_{4,n})$ when regarded as a hypersurface in $Z(x_0) = \bbP^n$. By \Cref{lem:sr-subspace}, this implies $\sr(\hat{S}_{4, n}) \leq 2$, in  contradiction with \Cref{lem:reduced-shioda-lowerbound}. This contradiction completes the proof.
\end{proof}

As a corollary we obtain a lower bound for homogeneous ABP size for Shioda's polynomials in six variables, which improves on Kumar's lower bound for a polynomial with the same number of variables.

\begin{corollary}
$\abphom(S_{4, d}) \geq 3(d - 1) + 2$.
\end{corollary}
\begin{proof}
We have $\str_k(S_{4, n)} \geq 3$ from \Cref{prop:strength-codim-sing} and $\str_{d - 1} = \str_1(S_{4, n}) = \sr(S_{4, n)} = 4$ from \Cref{thm:exactslicerank}.
Using \Cref{thm:abp-sum} we get the required lower bound.
\end{proof}

A similar (but easier) argument can be used to prove $\sr(S_{2, d}) = 3$. It follows that $\abphom(S_{2, d}) \geq 2(d - 1) + 2$.
We conjecture that this can be generalized to all Shioda polynomials.

\begin{conjecture}
\label{conj:lb}
For $n$ even we have $\sr(S_{n, d}) = \frac{n}{2} + 2$ and, consequently, 
\[
\abphom(S_{n, d}) \geq \frac{n + 2}{2}(d - 1) + 2.
\]
\end{conjecture}

\section{Geometry of algebraic branching programs}

We have seen that ABPs are related to degree-restricted strength decompositions, and degree-restricted strength decompositions are closely related to subvarieties of the hypersurface defined by the computed polynomial.
One can also connect existence of ABPs to subvarieties directly.

\begin{theorem}
Let $F$ be a homogeneous polynomial of degree $d$.
$F$ is computed by a homogeneous ABP with $w_k$ vertices in layer $k$
if and only if there exists a chain of ideals 
\[
I_1 \supset I_2 \supset \dots \supset I_{d - 1} \supset I_d = \left<F\right>
\]
such that $I_k$ is generated by $w_k$ homogeneous polynomials of degree $k$.
\end{theorem}
\begin{proof}
Suppose a homogeneous ABP $\mathsf{A}$ computes the polynomial $F$.
Recall that we denote the polynomial computed between vertices $v$ and $w$ by $\mathsf{A}[v, w]$.
Let $s$ and $t$ be the source and the sink of $\mathsf{A}$, and let $V_k$ be the set of vertices in the $k$-th layer.

Define $I_k$ to be the ideal generated by polynomials $\mathsf{A}[s, v]$ for all $v \in V_k$.
These polynomials are homogeneous degree $k$ polynomials, because every path from $s$ to $v \in V_k$ has length $k$.
If $w \in V_{k + 1}$, then
\begin{equation}\label{eq:abp-computation}
\mathsf{A}[s, w] = \sum_{v \in V_k} \mathsf{A}[s, v] \mathsf{A}[v, w],
\end{equation}
so all generators of $I_{k+1}$ lie in $I_k$ and thus $I_k \supset I_{k + 1}$.
The last layer of $\mathsf{A}$ contains only the sink, and the corresponding ideal is $\left< F \right>$.

On the other hand, given a sequence of ideals $I_1 \supset I_2 \supset \dots \supset I_{d - 1} \supset I_d = \left<F\right>$ such that $I_k$ is generated by homogeneous polynomials of degree $k$.
Let $G_{k1}, \dots, G_{kw_k}$ be the generators of $I_k$.
Since $I_j \supset I_{j + 1}$, we have
\begin{equation}\label{eq:abp-ideals}
G_{k+1, j} = \sum_{k = 1}^{w_j} G_{ki} L_{kij}
\end{equation}
for some linear forms $L_{kij}$.
Let $\mathsf{A}$ be an ABP with $w_k$ vertices in layer $k$ such that the edge from the $i$-th vertex in the $k$-th layer to the $j$-th vertex in the $(k+1)$-th layer is labeled by $L_{kij}$.
Then the equations~\eqref{eq:abp-computation} coincide with~\eqref{eq:abp-ideals} and thus the ABP computes the generator $F$ of the last ideal.
\end{proof}

Geometrically, this implies that $\Zero(F)$ contains a chain of subvarieties $X_1 \subset X_2 \subset \dots \subset X_{d - 1} \subset \Zero(F)$ where each $X_k$ is cut out by $w_k$ polynomials of degree $k$.

\bibliography{paper}

\end{document}